\documentclass{amsart}
\usepackage{graphicx}
\usepackage{dcolumn}
\usepackage{bm}
\usepackage{amsmath,amsthm,amssymb,stmaryrd}
\usepackage{enumerate}
\newtheorem{theorem}{Theorem}
\newtheorem{lemma}{Lemma}
\newtheorem{definition}{Definition}
\newtheorem{proposition}{Proposition}
\newtheorem{remark}{Remark}

\newcommand{\bbint}[2]{\ensuremath{\backslash\!\!\!\!\backslash\!\!\!\!\!\int_{#1}^{#2}}}

\begin{document}

\title[Analytic Principal Value]{The Cauchy Principal Value and the Hadamard finite part integral as values of absolutely convergent integrals}

\author{Eric A. Galapon}
\address{Theoretical Physics Group, National Institute of Physics, University of the Philippines, Diliman Quezon City, 1101 Philippines}
\email{eric.galapon@up.edu.ph}
\date{\today}
\maketitle

\begin{abstract}
	The divergent integral $\int_a^b f(x)(x-x_0)^{-n-1}\mathrm{d}x$, for $-\infty<a<x_0<b<\infty$ and $n=0, 1, 2, \dots$, is assigned, under certain conditions, the value equal to the simple average of the contour integrals $\int_{C^{\pm}} f(z)(z-x_0)^{-n-1}\mathrm{d}z$, where $C^+$ ($C^-$) is a path that starts from $a$ and ends at $b$, and which passes above (below) the pole at $x_0$. It is shown that this value, which we refer to as the Analytic Principal Value, is equal to the Cauchy principal value for $n=0$ and to the Hadamard finite-part of the divergent integral for positive integer $n$. This implies that, where the conditions apply, the Cauchy principal value and the Hadamard finite-part integral are in fact values of absolutely convergent integrals. Moreover, it leads to the replacement of the boundary values in the Sokhotski-Plemelj-Fox Theorem with integrals along some arbitrary paths. The utility of the Analytic Principal Value in the numerical, analytical and asymptotic evaluation of the principal value and the finite-part integral is discussed and demonstrated.
\end{abstract}

\section{Introduction}
Divergent integrals arise naturally in many areas of physics and engineering \cite{shiekh,lee,ioa2,lifanov}. In this paper we consider the class of non-converging integrals given by
\begin{equation}\label{divergent}
\int_a^b \frac{f(x)}{(x-x_0)^{n+1}} \mbox{d}x,\;\;\; -\infty < a<x_0<b<\infty,\;\; n=0,1,2,\dots 
\end{equation}
for some function $f(x)$ not vanishing at $x=x_0$. These have been assigned meaningful values by a symmetric removal of the singular point, $x=x_0$, of the integrand. In particular the integral, for a fixed $n$, is replaced with the limit
\begin{equation}\label{limit}
 \lim_{\epsilon\rightarrow 0^+}\left[\int_a^{x_0-\epsilon} \frac{f(x)}{(x-x_0)^{n+1}} \mathrm{d}x+ \int_{x_0+\epsilon}^{b} \frac{f(x)}{(x-x_0)^{n+1}} \mathrm{d}x\right],
\end{equation}
and a finite value is extracted which is assigned as the value of the divergent integral. Under some continuity conditions on $f(x)$, only the case $n=0$ leads to a well-defined limit, which is the well-known Cauchy Principal Value (CPV) \cite{pipkin}.  For positive integer $n=1,2, \dots$ the limit does not exist. However, expression \ref{limit} can be cast into a form with a group of terms possessing a finite value in the limit $\epsilon\rightarrow 0$ and into another group of terms that diverge in the same limit. The divergent integral is assigned a value by hand by dropping the diverging term, leaving the group of terms with a finite value in the limit, the limit of which is assigned as the value of the divergent integral \cite{ang,kanwal,cohen,chan,monegato}. This manner of assigning value to a divergent integral is due to Hadamard \cite{hadamard}, and the value is now known as the Hadamard finite part or the Finite-Part Integral (FPI). 

Historically it was Fox \cite{fox} who made the first investigation of the integral \ref{divergent} following earlier Hadamard's introduction of the finite-part of a divergent integral \cite{hadamard}. He obtained the explicit form of the Finite-Part Integral of equation \ref{divergent}, with the Cauchy Principal Value as a special case. It is given by 
\begin{eqnarray}\label{fox}
&&\#\!\!\int_a^b \frac{f(x)}{(x-x_0)^{n+1}} \mbox{d}x \nonumber \\
&&\hspace{12mm}=\lim_{\epsilon\rightarrow 0^+}\left[\int_a^{x_0-\epsilon} \frac{f(x)}{(x-x_0)^{n+1}} \mathrm{d}x+ \int_{x_0+\epsilon}^{b} \frac{f(x)}{(x-x_0)^{n+1}} \mathrm{d}x - H_n(x_0,\epsilon)\right],
\end{eqnarray}
where
\begin{eqnarray}
H_n(x_0,\epsilon)&=&0, \;\;\; n=0\label{fox1} \\
H_n(x_0,\epsilon)&=& \sum_{k=0}^{n-1} \frac{f^{(k)}(x_0)}{k!(n-k)} \frac{(1-(-1)^{n-k})}{\epsilon^{n-k}},\;\;\; n=1, 2, \dots  , \label{fox2}
\end{eqnarray}
in which $\#$ is the Cauchy Principal Value for $n=0$ and $\#$ is the Finite-Part Integral for positive integer $n$. Fox referred to equation \ref{fox} as the principal value, appropriately so as the CPV is a special case; however, we will continue to call equation \ref{fox} as the finite part integral, in accordance with current literature. The limit in equation \ref{fox} exists if $f(x)$ possesses derivatives up to order $n$ and satisfies a Holder condition. 

While the Cauchy principal value and the finite-part of the divergent  integrals given by expression \ref{divergent} have been studied extensively and have reached the status of an orthodoxy in applied mathematics \cite{lifanov,ang,pipkin,kanwal}, in this paper we point to another perspective on the integral which, we believe, should have had been much earlier appreciated but somehow had escaped the researchers of the field. The basic idea promoted here has already been intimated by Hadamard in considering a family of divergent integrals similar to equation \ref{divergent} \cite{hadamard}, but it was not developed further. Under certain conditions on the function $f(x)$, here we assign a value to the integral \ref{divergent} by the average of the values of the same integral when the contour of integration is displaced above and below the pole at $z=x_0$ while keeping the limits of the integration fixed. We will refer to this value as the Analytic Principal Value (APV), for a reason to made clear below. Our introduction of the APV is motivated by the Sokhotski-Plemelj theorem (SPT) for the CPV \cite{sokho,plemelj} and its generalization due to Fox \cite{fox}. According to SPT the CPV is the average of the boundary values of the analytic function, $\Phi(z)$, obtained from the integral \ref{divergent} (with $n=0$) by replacing $x_0$ with the complex variable $z$; the boundary values are the values of $\Phi(z)$ as $z$ approaches the singular point $x_0$ from above and below (the real-axis). 

In our definition of the Analytic Principal Value, we do the opposite in Sokhostski-Plemelj theorem. In the SPT the singular point $x_0$ is lifted out of the path of integration with the replacement $x_0\rightarrow z$, the path being fixed. Here the point $x_0$ is fixed but the contour of integration is deformed above and below the singularity, defining two families of paths separated by the pole, with each family having a unique value. The simple average of these values is our APV. We will show that the Analytic Principal Value is equal to the Cauchy Principal Value for $n=0$ and to the Finite-Part Integrals for positive integer $n$. The APV then obviates the need to introduce boundary values. In fact, the SPT and its generalization can be seen as a specific realization of the Analytic Principal Value; close scrutiny, for example, of a particular implementation of the limiting procedure in the SPT shows that it is no more than a special computation of the APV \cite{pipkin}. 

A feature of the APV is that its computation does not require the need for a limiting process. It assumes the form of an absolutely convergent integral, so that the CPV and the FPI are themselves values of absolutely convergent integrals.  That is the divergent integrals \ref{divergent}, interpreted as APVs, are convergent integrals in disguise. Beside its conceptual appeal, such integral representation of the CPV and the FPI make them amenable to computation using standard numerical quadratures, obviating the need for specialized algorithms specifically tailored for them \cite{cris,bia}. Furthermore, an (absolutely convergent) integral representation allows an asymptotic analysis of the CPV and the FPI \cite{guer} using the already established methods of asymptotic analysis for regular integrals \cite{wong}.

This paper is organized as follows. In Section-\ref{principalvalue} we formalize the definition of the Analytic Principal Value of the given family of divergent integrals. In Section-\ref{cpvfpi} we show that the Principal Value equals the Cauchy Principal Value and the Finite-Part Integral. In Section-\ref{spft} we discuss the relationship between our results and the  Sokhotski-Plemelj theorem for the CPV and its generalization due to Fox for the FPI. In Section-\ref{example} we apply the definition of the Analytic Principal Value to a specific case, and demonstrate the utility of the absolutely convergent integral representation of the CPV and the PFI. In Section-\ref{conclusion} we conclude and point to further applications of the Principal Value and raise an open problem.

\section{The Analytic Principal Value Integral}\label{principalvalue}

Given the function $f(x)$ in the divergent integral \ref{divergent}, let us introduce the complex valued function $f(z)$ obtained by replacing the real variable $x$ with the complex variable $z$ in $f(x)$. We will refer to $f(z)$ as the complex extension of $f(x)$. We require that there exists a neighborhood $R$ that encloses the strip $[a,b]$ and in this region $f(z)$ is analytic. $f(z)$ can have an infinite number of poles as long as none of them are in $R$. Denote the punctured domain by $R_{x_0}=R\ \backslash\{x_0\}$. Also denote $\Gamma^+$ the set of all continuous, non-self-intersecting paths contained in $R_{x_0}$ that start at $a$ and end at $b$, and that pass the pole at $z=x_0$ above the real axis; and $\Gamma^-$ the set of all similar paths in $R_{x_0}$  
that pass the pole below the real axis (see Figure-1). 

Due to the analyticity of $f(z)/(z-x_0)^{n+1}$ in the region $R_{x_0}$, the value of the integral
\begin{equation}\label{contour}
\int_{\gamma}\frac{f(z)}{(z-x_0)^{n+1}} \mathrm{d}z
\end{equation}
does not depend on $\gamma$ when $\gamma$ is restricted on $\Gamma^+$ or $\Gamma^-$ only. That is there is a single value to the integral \ref{contour} for all paths $\gamma^+$ in $\Gamma^+$, which we denote by $\mathrm{Int}^+(x_0)$; similarly for all paths $\gamma^-$ in $\Gamma^-$, which we denote by $\mathrm{Int}^-(x_0)$. However, due to the pole at $z=x_0$, the value of the integrals $\mathrm{Int}^+(x_0)$ and $\mathrm{Int}^-(x_0)$ are generally not equal. We now define the Analytic Principal Value of the divergent integral \ref{divergent} as follows.
\begin{definition} 
	Let $f(x)$ admit a complex extension $f(z)$ that is analytic in a neighborhood containing the strip $[a,b]$. Then the Analytic Principal Value of the divergent integral, to be denoted by ${\scriptstyle \backslash\!\!\!\!\backslash\!\!\!\!} { \int}$,  is given by the simple average of $\mathrm{Int}^+(x_0)$ and $\mathrm{Int}^-(x_0)$,
\begin{equation}
\backslash\!\!\!\!\backslash\!\!\!\!\!\int_a^b \frac{f(x)}{(x-x_0)^{n+1}}\,\mathrm{d}x =\frac{1}{2}\left[ \mathrm{Int}^+(x_0) + \mathrm{Int}^-(x_0)\right]\label{principal} .
\end{equation}
\end{definition}

\begin{lemma} For $n=0, 1, 2, \dots$
	 \begin{equation}\label{relation}
	\mathrm{Int}^-(x_0) - \mathrm{Int}^+(x_0)= 2 \pi i \frac{f^{(n)}(x_0)}{n!},
	\end{equation} 
 \begin{equation}\label{form1}
\backslash\!\!\!\!\backslash\!\!\!\!\!\int_a^b \frac{f(x)}{(x-x_0)^{n+1}}\,\mathrm{d}x = \mathrm{Int}^+(x_0) + i \pi \frac{f^{(n)}(x_0)}{n!}.
\end{equation}
 \begin{equation}\label{form2}
\backslash\!\!\!\!\backslash\!\!\!\!\!\int_a^b \frac{f(x)}{(x-x_0)^{n+1}}\,\mathrm{d}x = \mathrm{Int}^-(x_0) - i \pi \frac{f^{(n)}(x_0)}{n!}.
\end{equation}
\end{lemma}
\begin{proof}
	Equation \ref{relation} follows from the residue theorem and the fact that the closed path $\gamma^-+(-\gamma^+)$, for every $\gamma^{\pm}\in\Gamma^{\pm}$, encloses the pole $z=x_0$. Equations \ref{form1} and \ref{form2} are consequences of the definition of the Principal Value given by equation \ref{principal} and the relationship between $\mathrm{Int}^+(x_0)$ and $\mathrm{Int}^-(x_0)$ given by equation \ref{relation}.
\end{proof}

Equations \ref{principal}, \ref{form1} and \ref{form2} are equivalent and any one of them can be used to compute the Analytic Principal Value. Equation \ref{principal} requires two paths, while equations \ref{form1} and \ref{form2} require one path each. The integrals involved in these expressions are absolutely convergent, so that the family of divergent integrals given by equation \ref{divergent} are in fact absolutely convergent integrals when interpreted as Analytic Principal Values. But since the choice of path or paths in $\Gamma^{\pm}$ is arbitrary, the divergent integrals assume various but equivalent (convergent) integral representations. Any one of these representation can be chosen to conveniently compute the principal value, either analytically or numerically.

The above definition of the Analytic Principal Value extends readily when the path lies in the complex plane. 

\begin{figure}
	\includegraphics[scale=0.4]{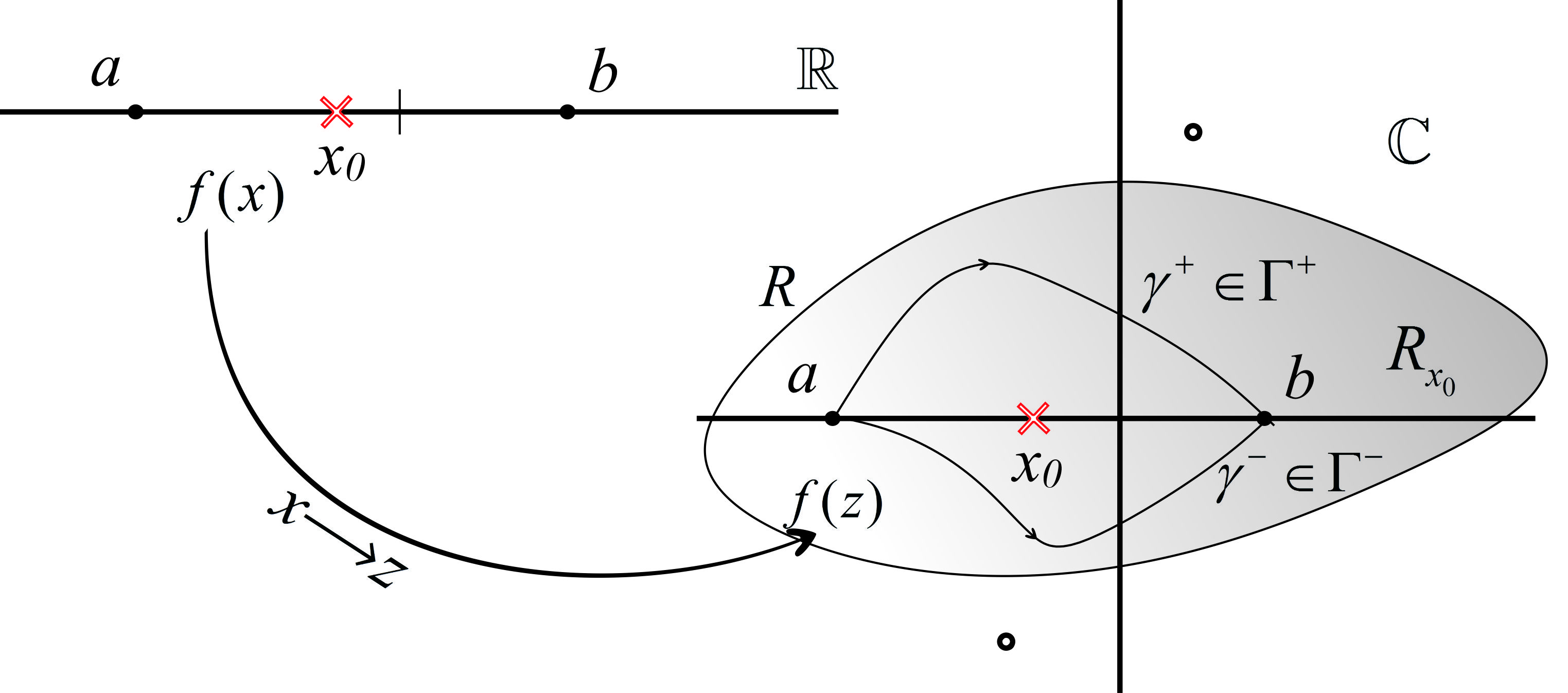}
	\caption{The mapping of the region of integration from the real line to a region in the complex plane. If $f(z)$ has poles, the region $R$ is chosen such that the poles, indicated by the hollow circles, are outside $R$.}
	\label{fig:boat1}
\end{figure}

\section{The Cauchy Principal Value and the Finite Part Integrals as Values of the Analytic Principal Value}\label{cpvfpi}
\begin{proposition} For $n=0$ the Analytic Principal Value is equal to the Cauchy principal value,
	\begin{equation}\label{cpv}
	\bbint{a}{b}\frac{f(x)}{x-x_0} \mathrm{d}x = \lim_{\epsilon\rightarrow 0}\left[\int_a^{x_0-\epsilon} \frac{f(x)}{x-x_0}\mathrm{d}x + \int_{x_0+\epsilon}^b\frac{f(x)}{x-x_0}\mathrm{d}x\right] .
	\end{equation}
\end{proposition}
\begin{proof}
	To compute the analytic principal value, it is sufficient to chose a convenient path in $\Gamma^+$ and use equation \ref{form1}. We use the contour $\bar{\gamma}^+$ shown in Figure-2. Then
	\begin{equation}\label{precpv}
	\bbint{a}{b} \frac{f(x)}{x-x_0}\mathrm{d}x=\int_a^{x_0-\epsilon} \frac{f(x)}{x-x_0}\mathrm{d}x + \int_{x_0+\epsilon}^b\frac{f(x)}{x-x_0}\mathrm{d}x + \int_{C^+}\frac{f(z)}{z-x_0}\mathrm{d}z + i\pi f(x_0)
	\end{equation} 
	where $C^+$ is the semi-circle centered at the origin with radius $\epsilon$. We only need to evaluate the integral along $C^+$. We parametrize the semi-circle by $z=x_0+\epsilon \mathrm{e}^{i\theta}$, where $\pi>\theta>0$. Because $f(z)$ is analytic in $R$, we can use the Taylor expansion theorem to expand $f(x_0+\epsilon \mathrm{e}^{i\theta})$ about $z=x_0$ up to the order $O(\epsilon)$, 
	\begin{equation}
	f(x_0+\epsilon\mathrm{e}^{i\theta})=f(x_0) + \epsilon \mathrm{e}^{i\theta} f_1(x_0+\epsilon\mathrm{e}^{i\theta}),
	\end{equation}
	where $f_1(z)$ is an analytic function \cite{ahlfors}. From the Taylor expansion theorem, we have the bound
	\begin{equation}\label{bound}
	\left| f_1(x_0+\epsilon\mathrm{e}^{i\theta})\right|\leq \frac{M}{\rho(\rho-\epsilon)},
	\end{equation}
	where $M$ is the maximum of $|f(z)|$ in $R$, and $\rho$ is any positive constant with $\rho>\epsilon$ and such that the disk $|z-x_0| \leq \rho$  is contained in $R$. We can choose $\epsilon$ sufficiently small to satisfy the last requirement. Then we have
	\begin{equation}
	\int_{C^+}\frac{f(z)}{z-x_0}\mathrm{d}z= -i\pi f(x_0) +\epsilon i \int_{\pi}^{0} \mathrm{e}^{i\theta} f_1(x_0+\epsilon\mathrm{e}^{i\theta}) \mathrm{d}\theta .
	\end{equation}

	We substitute this back into equation \ref{precpv} and the residue term cancels out. We obtain
	 \begin{equation}\label{precpv2}
	 \bbint{a}{b} \frac{f(x)}{x-x_0}\mathrm{d}x=\int_a^{x_0-\epsilon} \frac{f(x)}{x-x_0}\mathrm{d}x + \int_{x_0+\epsilon}^b\frac{f(x)}{x-x_0}\mathrm{d}x +\epsilon i \int_{\pi}^{0}\mathrm{e}^{i\theta} f_1(x_0+\epsilon\mathrm{e}^{i\theta}) \mathrm{d}\theta  .
	 \end{equation} 
	 Using the bound given by \ref{bound} we obtain the inequality 
	 \begin{eqnarray}\label{ineq1}
	 &&\left|\bbint{a}{b} \frac{f(x)}{x-x_0}\mathrm{d}x-\left[\int_a^{x_0-\epsilon} \frac{f(x)}{x-x_0}\mathrm{d}x + \int_{x_0+\epsilon}^b\frac{f(x)}{x-x_0}\mathrm{d}x \right] \right|\nonumber \\
	 &&\hspace{44mm}\leq \epsilon \int^{\pi}_{0} |f_1(x_0+\epsilon\mathrm{e}^{i\theta})| \mathrm{d}\theta \leq \frac{\pi M \epsilon}{\rho(\rho-\epsilon)}
	 \end{eqnarray}
	 Since the right hand side of the inequality \ref{ineq1} becomes arbitrarily small for arbitrarily small $\epsilon$, we obtain the equality \ref{cpv} in the limit as $\epsilon$ approaches zero.
\end{proof}

\begin{proposition}
	For positive integer $n$, the Analytic Principal Value is equal to the finite-part of the divergent integral,
	\begin{eqnarray}\label{general}
	\bbint{a}{b} \frac{f(x)}{(x-x_0)^{n+1}}\mathrm{d}x&=&\lim_{\epsilon\rightarrow 0}\left[\int_a^{x_0-\epsilon} \frac{f(x)}{(x-x_0)^{n+1}}\mathrm{d}x + \int_{x_0+\epsilon}^b\frac{f(x)}{(x-x_0)^{n+1}}\mathrm{d}x \right. \nonumber \\
	&& \left.\hspace{18mm} -\sum_{k=0}^{n-1}\frac{f^{(k)}(x_0)}{k! (n-k)} \frac{(1-(-1)^{n-k})}{\epsilon^{n-k}} \right] .
	\end{eqnarray}
\end{proposition}
\begin{proof}
We use the same contour of integration to calculate the analytic principal value and obtain
	\begin{eqnarray}\label{ddd}
	\bbint{a}{b} \frac{f(x)}{(x-x_0)^{n+1}}\mathrm{d}x&=&\int_a^{x_0-\epsilon} \frac{f(x)}{(x-x_0)^{n+1}}\mathrm{d}x + \int_{x_0+\epsilon}^b\frac{f(x)}{(x-x_0)^{n+1}}\mathrm{d}x \nonumber \\
	&&\hspace{12mm}+ \int_{C^+}\frac{f(z)}{(z-x_0)^{n+1}}\mathrm{d}z + i\pi \frac{f^{(n)}(x_0)}{n!} .
	\end{eqnarray}
	To evaluate the integral around the semi-circle, we again parametrize the semi-circle in the same way we did above and expand $f(x_0+\epsilon \mathrm{e}^{i\theta})$ at least up to the order $O(\epsilon^{n+1})$,
	\begin{equation}\label{expand2}
	f(x_0+\epsilon \mathrm{e}^{i\theta}) = \sum_{k=0}^n \frac{f^{(k)}(x_0)}{k!} \epsilon^k \mathrm{e}^{i k \theta} + \epsilon^{n+1} \mathrm{e}^{i\theta (n+1) } f_{n+1}(x_0+\epsilon\mathrm{e}^{i\theta}),
	\end{equation} 
	where $f_{n+1}(z)$ is an analytic function. Again from the Taylor expansion theorem, we have the bound
	\begin{equation}
	| f_{n+1}(x_0+\epsilon\mathrm{e}^{i\theta})|\leq \frac{M}{\rho^n (\rho-\epsilon)},
	\end{equation}
	where $M$ and $\rho$ are as above. 	Substituting the expansion \ref{expand2} back into equation \ref{ddd} and performing the integrations, we obtain
	\begin{eqnarray}\label{ddd2}
	\bbint{a}{b} \frac{f(x)}{(x-x_0)^{n+1}}\mathrm{d}x&=&\int_a^{x_0-\epsilon} \frac{f(x)}{(x-x_0)^{n+1}}\mathrm{d}x + \int_{x_0+\epsilon}^b\frac{f(x)}{(x-x_0)^{n+1}}\mathrm{d}x \nonumber \\
	&&\hspace{-8mm} -\sum_{k=0}^{n-1}\frac{f^{(k)}(x_0)}{k! (n-k)} \frac{(1-(-1)^{n-k})}{\epsilon^{n-k}} + i\epsilon \int_{\pi}^{0}\mathrm{e}^{i\theta} f_{n+1}(x_0+\epsilon\mathrm{e}^{i\theta}) \mathrm{d}\theta.
	\end{eqnarray}
	By the same arguments we used in the previous Proposition, we obtain equation \ref{general} from equation \ref{ddd2} in the limit as $\epsilon$ approaches zero. 
\end{proof}

\begin{figure}
	\includegraphics[scale=0.4]{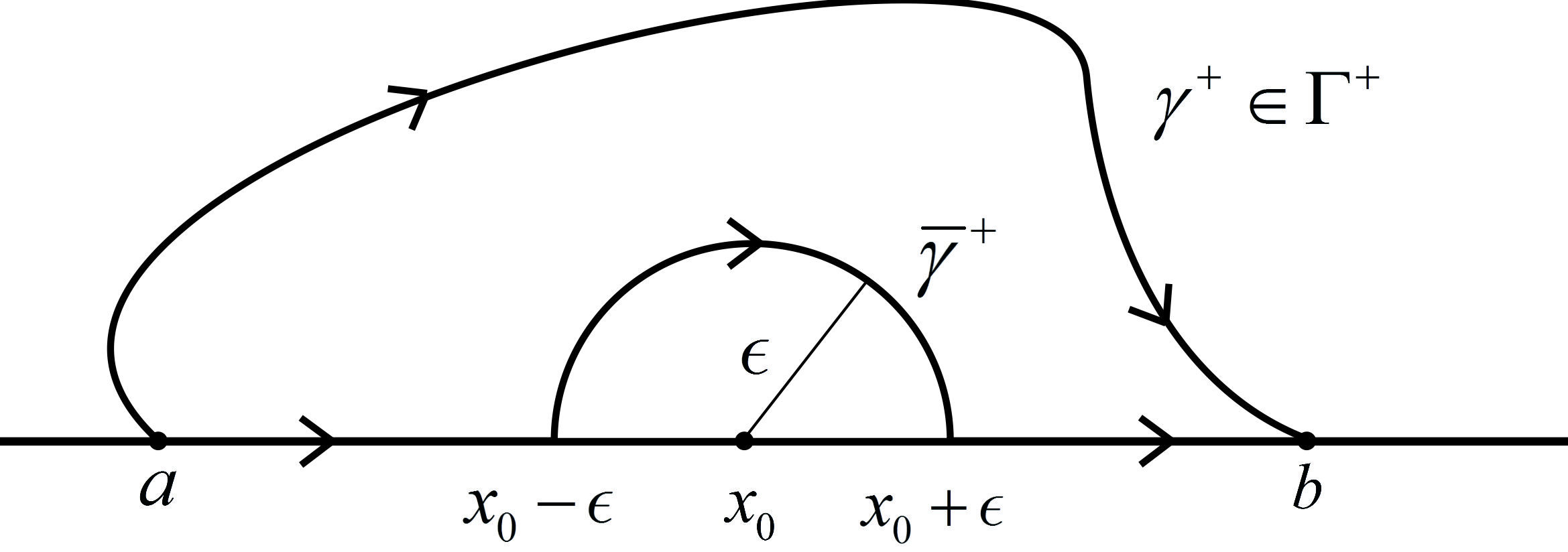}
	\caption{The contour of integration $\gamma^+$ is deformed into the contour $\bar{\gamma}^+$. The value of the contour integral along these paths are equal, in particular, the integral does not depend on $\epsilon$.}
	\label{fig:boat2}
\end{figure}

As pointed out earlier, the APV assumes an absolutely convergent integral representation. This, together with Propositions-1 and-2, leads us to the following statement of our main result.
\begin{theorem} 
	If the function $f(x)$ in the divergent integral \ref{divergent} has an analytic complex extension, $f(z)$, in a region in the complex plane containing the integration interval $[a,b]$, then there exists a family of absolutely convergent integral with a value equal to the Cauchy principal value for $n=0$ or to the finite part integral for $n=1, 2, \dots$. The family of integrals is precisely the different integral representations of the Analytic Principal Value of the divergent integral. Each integral representation corresponds to a path in $\Gamma^{\pm}$ or to a pair of paths, one from $\Gamma^+$ and another from $\Gamma^-$. 
\end{theorem}

\section{The Sokhotski-Plemelj-Fox Theorem and the Analytic Principal Value}\label{spft}
Lemma-1 looks exactly as the Sokhotski-Plemelj theorem (for $n=0$) and its generalized version by Fox. However, they are not the same, but their similarity in form indicates an existing relationship between our Lemma and them. 

Let us consider the version of the Sokhotski-Plemelj and Fox Theorem in the real line, which we will refer to as the Sokhotski-Plemelj-Fox Theorem (SPFT). Central to the SPFT is the function 
\begin{equation}
	\Phi(z)=\int_a^b \frac{f(x)}{(x-z)^{n+1}}\mathrm{d}x ,
\end{equation}
where $z$ does not lie along the contour of integration which is the straight line from $a$ to $b$. The value of $\Phi(z)$ as it approaches any value $x_0\in(a,b)$ depends on the approach, depending on whether the limit is from above or from below the real axis. The limiting values are known as the boundary values of $\Phi(z)$ and they are given by the limiting values
\begin{equation}
	\Phi^{\pm}(x_0)=\lim_{y\rightarrow 0} \int_a^b \frac{f(x)}{(x-(x_0\pm iy))^{n+1}}\mbox{d}x .
\end{equation}

The Sokhotski-Plemelj-Fox Theorem is a statement on the relationship between these values and the Cauchy Principal Value or the Finite-Part Integrals, in particular
\begin{equation}\label{boundary1}
	\Phi^{+}(x_0)= \#\!\!\!\int_a^b \frac{f(x)}{(x-x_0)^{n+1}} \mbox{d}x + i \pi \frac{f^{(n)}(x_0)}{n!} .
\end{equation}
\begin{equation}\label{boundary2}
	\Phi^{-}(x_0)= \#\!\!\!\int_a^b \frac{f(x)}{( x-x_0)^{n+1}} \mbox{d}x - i \pi \frac{f^{(n)}(x_0)}{n!} ,
\end{equation}
where $\#$ takes on either the CPV ($n=0$) or the FPI ($n=1,2,\dots$). Adding this two gives
\begin{equation}\label{plem}
	\#\!\!\!\int_a^b \frac{f(x)}{(x-x_0)^{n+1}} \mbox{d}x = \frac{1}{2}\left[\Phi^+(x_0)+\Phi^-(x_0)\right] .
\end{equation}
That is the CPV and the FPI are the averages of the boundary values of $\Phi(z)$.

We observe that there is an exact correspondence between equations \ref{principal} and \ref{plem}, equations \ref{form1} and \ref{boundary1}, equations \ref{form2} and \ref{boundary2}. Since the Analytic Principal Value and the Cauchy Principal Value/Finite-Part Integrals are equal, we must have the equality
\begin{eqnarray}
	\Phi^{\pm}(x_0)=\mathrm{Int}^{\mp}(x_0) .
\end{eqnarray}
That is the value of the function $\Phi(z)$ at the boundary is equal to the integral on any path connecting $a$ and $b$ that is deformable to the cut, with the value depending on which side of the cut the path passes through. Then the boundary value in the SPT can be replaced by a contour integral. Since the boundary values $\Phi^{\pm}(x_0)$ can be interpreted as values for particular paths, the SPFT can be seen as a special case of Lemma-1. 

Of course we can still maintain the interpretation of the Sokohotski-Plemelj-Fox Theorem as a statement on the relationship between the boundary values of $\Phi(z)$ and the Cauchy Principal Value or the Finite-Part Integral. With this interpretation, the SPFT stands independent from our Lemma. However, Lemma-1, together with Propositions-1 and 2, now provides a way of computing for the boundary values without explicit evaluation of the function $\Phi(z)$ and then taking the required limit. Conversely we can take the SPFT to be a means of computing the integral values $\mathrm{Int}^{\pm}(x_0)$ in terms of the boundary values of $\Phi(z)$. Either way we have now more ways of obtaining the Cauchy Principal Value, and the Finite-Part Integral.

\begin{remark}
	The proof of Proposition-1 is exactly the implementation of the limiting operation in Sokhotski-Plemelj Theorem in obtaining the boundary values as performed in \cite{pipkin}. It can now be seen that such implementation is more properly interpreted as a specific computation of the Analytic Principal Value rather than the computation of the boundary value coming from the function $\Phi(z)$ itself; that is because the contour is supposed to be fixed in the SPT, and it is the complex variable $z$ that is supposed to approach the singular point $x_0$. 
\end{remark}

\section{The Cauchy principal value and the finite-part integral involving functions with entire complex extensions}\label{example}
Let us consider the case when  the complex extension, $f(z)$, is entire, so that it posses a Taylor series expansion at any point with an infinite radius of convergence. In particular the following expansion holds
\begin{equation}\label{ps}
	f(z)=\sum_{k=0}^{\infty} \frac{f^{(k)}(x_0)}{k!} (z-x_0)^k
\end{equation}
for all $z$ in the complex plane. This gives us the opportunity to evaluate explicitly the CPV and the FPI and compare them with the APV. Moreover, it will be instructive to demonstrate the independence of the APV on the parameter $\epsilon$.

\subsection{The Cauchy Principal Value and the Finite-Part Integral} The CPV and the FPI are known, but we give their derivations here to aid in the derivation of the Analytic Principal Value. Term by term integration, which we can do because the expansion \ref{ps} has an infinite radius of convergence, yields the following integrals,
\begin{eqnarray}
	\int_a^{x_0-\epsilon} \frac{f(x)}{(x-x_0)^{n+1}}\mbox{d}x &=& -\sum_{k=0}^{n-1} \frac{f^{(k)}(x_0)}{k!(n-k)} \left(\frac{1}{(-\epsilon)^{n-k}} - \frac{1}{(a-x_0)^{n-k}}\right) \nonumber \\
	&& + \frac{f^{(n)}(x_0)}{n!} \left(\ln \epsilon - \ln(x_0-a)\right)\nonumber \\
	&& + \sum_{k=n+1}^{\infty} \frac{f^{(k)}(x_0)}{k!(k-n)} \left((-\epsilon)^{k-n}-(a-x_0)^{k-n}\right)
\end{eqnarray}
\begin{eqnarray}
	\int_{x_0+\epsilon}^b \frac{f(x)}{(x-x_0)^{n+1}}\mbox{d}x &=& -\sum_{k=0}^{n-1} \frac{f^{(k)}(x_0)}{k!(n-k)} \left(\frac{1}{(b-x_0)^{n-k}} - \frac{1}{\epsilon^{n-k}}\right) \nonumber \\
	&& + \frac{f^{(n)}(x_0)}{n!} \left(\ln (b-x_0) - \ln\epsilon\right)\nonumber \\
	&& + \sum_{k=n+1}^{\infty} \frac{f^{(k)}(x_0)}{k!(k-n)} \left((b-x_0)^{k-n}-\epsilon^{k-n}\right)
\end{eqnarray}

Adding these two gives the desired integral in the extraction of the finite-part of the divergent integral,
\begin{eqnarray}\label{desired}
	&&\int_a^{x_0-\epsilon} \frac{f(x)}{(x-x_0)^{n+1}} \mathrm{d}x+ \int_{x_0+\epsilon}^{b} \frac{f(x)}{(x-x_0)^{n+1}} \mathrm{d}x \nonumber \\
	&&\hspace{10mm}= \frac{f^{(n)}(x_0)}{n!}\left[\ln(b-x_0)-\ln(x_0-a)\right]+\left[F_n(b-x_0)-F_n(a-x_0)\right]\nonumber\\
	&&\hspace{40mm} - \left[F_n(\epsilon) - F_n(-\epsilon)\right]
\end{eqnarray}
where
\begin{equation}\label{fn}
	F_n(s)=-\sum_{k=0}^{n-1} \frac{f^{(k)}(x_0)}{k! (n-k)} \frac{1}{s^{n-k}}   + \sum_{k=n+1}^{\infty} \frac{f^{(k)}(x_0)}{k! (k-n)} s^{k-n} .
\end{equation}

The first two terms in equation \ref{desired} are independent of $\epsilon$, and the third term contains a group of terms that diverge as $\epsilon\rightarrow 0$. For $n=0$ the diverging group of term in \ref{desired} are not present. Then in the limit we obtain the Cauchy principal value
\begin{eqnarray}
\mathrm{CPV}\!\!\int_a^b \frac{f(x)}{x-x_0} \mathrm{d}x\!\!\!&=&\!\!\! \sum_{k=1}^{\infty} \frac{f^{(k)}(x_0)}{k! \times k} \left[(b-x_0)^k-(a-x_0)^k\right] \nonumber \\
&&\hspace{12mm} +  f(x_0) \left[\ln(b-x_0)-\ln(x_0-a)\right]\label{CPV}
\end{eqnarray} 
For $n=1, 2, \dots$ the diverging terms are present to make the integral diverge in the limit. The finite part is obtained by dropping the diverging term and is given by
\begin{eqnarray}
	\mbox{FPI}\!\!\int_a^b\frac{f(x)}{(x-x_0)^{n+1}}\mbox{d}x\!\!\!& =&\!\!\! F_n(b-x_0)-F_n(a-x_0)\nonumber\\
	&&\hspace{12mm} +  \frac{f^{(n)}(x_0)}{n!} \left[\ln(b-x_0)-\ln(x_0-a)\right].\label{FPI}
\end{eqnarray}
The group of terms that we dropped is precisely the sum given by equation \ref{fox2}. 

\subsection{The Analytic Principal Value}
To obtain the APV, let us first consider the integral in $\Gamma^+$. We deform the path $\gamma^+$ into the path consisting of the straight path from $a$ to $x_0-\epsilon$, and the semi-circle with radius $\epsilon$ centered at the origin in the positive direction, and then the straight path from $x_0+\epsilon$ to $b$ for some $\epsilon>0$, as depicted in Figure-2. Then the integral in the upper half-plane becomes
\begin{eqnarray}
\int_{\gamma^+} \frac{f(z)}{(z-x_0)^{n+1}} \mathrm{d}z &=& \int_a^{x_0-\epsilon} \frac{f(x)}{(x-x_0)^{n+1}} \mathrm{d}x  +\int_{x_0+\epsilon}^{b} \frac{f(x)}{(x-x_0)^{n+1}} \mathrm{d}x \nonumber \\
&& \hspace{20mm} + \int_{C^+} \frac{f(z)}{(z-x_0)^{n+1}} \mathrm{d}z
\end{eqnarray}
Notice that the first and second terms are just the integrals appearing in the definition of the finite part integral.  

The integral around the semi-circle is performed with the parametrization $z=x_0+ \epsilon \mathrm{e}^{i\theta}$, with $\pi>\theta>0$. Expanding $f(z)$ along the contour of integration
\begin{equation}
f\left(x_0+\epsilon \mathrm{e}^{i\theta}\right)=\sum_{k=0}^{\infty} \frac{f^{(k)}(x_0)}{k!} \epsilon^{k} \mathrm{e}^{i k \theta} ,
\end{equation}
and substituting its expansion in the integral yield 
\begin{eqnarray}
\int_{C^+} \frac{f(z)}{(z-x_0)^{n+1}} \mbox{d}z &=& -\sum_{k=0}^{n-1} \frac{f^{(k)}(x_0)}{k!(n-k)} \left(\frac{1}{\epsilon^{n-k}} - \frac{1}{(-\epsilon)^{n-k}}\right) - i \pi \frac{f^{(n)}(x_0)}{n!}\nonumber \\
&& + \sum_{k=n+1}^{\infty} \frac{f^{(k)}(x_0)}{k!(n-k)} \left(\epsilon^{n-k}-(-\epsilon)^{n-k}\right) .
\end{eqnarray}
Adding all the terms, we find terms involving $\epsilon$ cancel out, leaving the integral independent of $\epsilon$, as it should be. We obtain
\begin{eqnarray}\label{upper}
\int_{\gamma^+} \frac{f(z)}{(z-x_0)^{n+1}}\mbox{d}x\!\!\!& =&\!\!\! F_n(b-x_0)-F_n(a-x_0) \nonumber \\
&&\hspace{1mm}+ \frac{f^{(n)}(x_0)}{n!} \left[\ln(b-x_0)-\ln(x_0-a)\right] - i \pi \frac{f^{(n)}(x_0)}{n!} .
\end{eqnarray}

To obtain the integral in $\Gamma^-$, we use the contour which is the mirror image of the contour in the evaluation of the integral in $\Gamma^+$. Following similar steps, we obtain the integral in the lower half-plane
\begin{eqnarray}\label{lower}
\int_{\gamma^-} \frac{f(z)}{(z-x_0)^{n+1}}\mbox{d}x\!\!\!& =&\!\!\! F_n(b-x_0)-F_n(a-x_0) \nonumber \\
&&\hspace{1mm}+ \frac{f^{(n)}(x_0)}{n!} \left[\ln(b-x_0)-\ln(x_0-a)\right] + i \pi \frac{f^{(n)}(x_0)}{n!} .
\end{eqnarray}
Observe that equations \ref{upper} and \ref{lower} differ only in sign in the second term. which is due to the fact that the two paths are oppositely oriented.

While equations \ref{upper} and \ref{lower} are computed using particular paths, their values, however, do not depend on the chosen path. Their right-hand sides give the values $\mathrm{Int}^+(x_0)$ and $\mathrm{Int}^-(x_0)$, respectively. Any of the equations \ref{principal}, \ref{form1} and \ref{form2} can now be used to reproduce the Cauchy principal value and the finite-part integrals obtained above.

\subsection{Example} 
Let us compute the Cauchy principal value and the finite-part integral of the following divergent integral,
\begin{equation}
\int_{-1}^{1} \frac{\cos x}{x^{n+1}} \mathrm{d}x, \;\;\; n=0,1,\dots ,
\end{equation}
using equations \ref{CPV} and \ref{FPI}, respectively. It is evident from equation \ref{CPV} that the CPV vanishes. Also the FPI vanishes for even $n$. Only when $n$ is odd that the FPI does not vanish. Using equation \ref{FPI}, the first two non-zero FPI are determined to be
\begin{eqnarray}
\mathrm{FPI}\!\!\int_{_-1}^{1} \frac{\cos x}{x^2}\mathrm{d}x \!\!\!&=&\!\!\! -2 + 2\sum_{j=1}^{\infty} \frac{(-1)^j}{(2j)! (2j-1)}\nonumber \\
\!\!\!&=&\!\!\! - 2 \left[\cos(1)+\mathrm{Si}(1)\right],\label{fpi1}
\end{eqnarray}
\begin{eqnarray}
	\mathrm{FPI}\!\!\int_{_-1}^{1} \frac{\cos x}{x^4}\mathrm{d}x \!\!\!&=&\!\!\! \frac{1}{3} +2 \sum_{j=2}^{\infty} \frac{(-1)^j}{(2j)! (2j-3)}\nonumber \\
	\!\!\!&=&\!\!\! \frac{1}{3} \left[\mathrm{Si}(1)+\sin(1)-\cos(1)\right] \label{fpi2},
\end{eqnarray}
where $\mathrm{Si}(z)$ is the sine-integral function.

Now we obtain the Analytic Principal Value of the same divergent integrals. The complex extension of $f(x)=\cos x$ is $f(z)=\cos z$, which is an entire function. Let us use the definition of the Principal Value given by equation \ref{principal}. We chose for $\gamma^+$ the semi-circle in the upper-half plane centered at the origin with a unit radius; and for $\gamma^-$ the semi-circle in the lower half plane centered at the origin as well. Parameterizing the paths as $z=a \mathrm{e}^{i\theta}$, where $\pi<\theta<0$ for $\gamma^+$ and $-\pi<\theta<0$ for $\gamma^-$, and averaging the integrals for these paths yield the APV
\begin{eqnarray}\label{pvex}
\bbint{-1}{1} \frac{\cos x}{x^{n+1}}\mathrm{d}x &=& -\int_0^{\pi}\sin(\cos\theta) \sinh(\sin\theta) \cos(n\theta) \mathrm{d}\theta \nonumber \\
&& -\int_0^{\pi}\cos(\cos\theta) \cosh(\sin\theta) \sin(n\theta) \mathrm{d}\theta 
\end{eqnarray}
The right hand side of the equation now involves an integration that is absolutely convergent which is amenable to standard methods of integration, either analytically or numerically. For even $n$ the principal value vanishes. Moreover, equations \ref{fpi1} and \ref{fpi2} are reproduced from equation \ref{pvex} by explicit evaluation of the integral (which we did using Mathematica 10.3), in accordance with the equality of the Analytic Principal Value and the Finite-Part Integral.

We may be interested in obtaining the behavior of the FPI for arbitrarily large $n$. It is not immediately clear from the definition of the FPI or from the power series how the asymptotic expansion can be obtained. This is where the integral representation comes in handy in obtaining the asymptotic expansion. The integral representation given by equation \ref{pvex} suggests that the expansion can be obtained by the standard method of integration by parts for Fourier integrals \cite{wong}. Successive integration by parts yields the expansion
\begin{eqnarray}
\bbint{-1}{1} \frac{\cos x}{x^{n+1}}\mathrm{d}x\!\!\! &\sim& \!\!\!((-1)^n-1) \left[\frac{\cos(1)}{n}-\frac{\cos(1)+\sin(1)}{n^3}+ \frac{5\sin(1)-6\cos(1)}{n^5} - \cdots \right] \nonumber \\
&& \hspace{-14mm}+ ((-1)^n-1) \left[-\frac{\sin(1)}{n^2}+\frac{3\cos(1)}{n^4} - \frac{5\cos(1)-23\sin(1)}{n^6} + \dots \right],\;\;\; n\rightarrow \infty .
\end{eqnarray}
The expansion identically vanishes for even $n$, which is consistent with the fact that the principal value vanishes for such values of $n$. This example demonstrates how the entire repertoire of asymptotic analysis for regular integrals can be used in the asymptotic analysis of hypersingular integrals using their (absolutely convergent) integral representations.

\section{Conclusion}\label{conclusion}
 We have shown that, under certain conditions, the Cauchy Principal Value and the Finite-Part Integral are values of absolutely convergent integrals. We have seen that such integral representation of them provides another way of computing their values, analytically, numerically or asymptotically, using standard methods applicable to regular integrals. Furthermore, it offers another way to look at problems at a different perspective. A convergent integral representation may allow us to cast, for example, integral equations involving singular kernels into integral equations involving regular kernels. But of course this cannot be done without complication: the domain of the unknown function will now have to be extended beyond the original domain. It is possible though that new insight can be gained from such rewriting of the original problem. Conversely our results here may allow us to rewrite integral equations involving regular kernels into hypersingular integral equations which, too, may offer new insights not available in the original formulation of the problem.  

Clearly there is an enormous potential in absolutely convergent integral representations of divergent integrals such as the family of integrals considered here. However, it is not apparent at the moment if such representation exists for every divergent integral. For our present case, such representation is possible under analyticity condition on the complex extension of the relevant function. However, such condition is a stringent one. It is only necessary for the function to satisfy a Holder condition for the CPV to exist; and for it to further posses derivatives up to order $n$ for the FPI to exist as well. Our results here then do not cover all possible cases of the CPV and the FPI; it is an open question whether an absolutely convergent integral representation exist for the rest of the cases. We leave it to future developments in the exploration of the possibility of obtaining an absolutely convergent integral representation of any given divergent integral.

\section*{Acknowledgement}
This work was funded by the UP System Enhanced Creative Work and Research Grant (ECWRG 2015-2-016).

\end{document}